\documentclass[12pt,reqno]{amsart}
\usepackage{amsmath,amssymb,amsthm,psfrag}
\usepackage{epsfig,graphicx,verbatim,bm}
\usepackage{color,psfrag}

\RequirePackage[colorlinks,citecolor=blue,urlcolor=blue]{hyperref}
\usepackage{breakurl}

\RequirePackage[colorlinks,citecolor=blue,urlcolor=blue]{hyperref}
\numberwithin{equation}{section}
\newtheorem{thm}[equation]{Theorem}

\newtheorem{lem}[equation]{Lemma}
\newtheorem{defn}[equation]{Definition}
\newcounter{mycount}

\newenvironment{letlist}{\begin{list}{(\alph{mycount})}%
   {\usecounter{mycount}\labelwidth=1cm\itemsep 0pt}}{\end{list}}

\newcommand{\tr}{\operatorname{tr}}

\newcommand{\uinvnorm}{|\kern-2pt|\kern-2pt|}

\def\1{\mbox{1\hskip-.25em l}}
\newcommand{\beq}{\begin{equation}}
\newcommand{\eeq}{\end{equation}}
\newcommand{\ZZ}{\mathbb{Z}}

\theoremstyle{plain}

\newtheorem{remark}[equation]{Remark}
\theoremstyle{definition}

\def\anb{a_{n,\b}}
\def\amb{a_{m,\b}}

\def\sH{\mathcal H}

\def\qq{\qquad}
\def\q{\quad}
\def\a{\alpha}
\def\b{\beta}
\def\de{\delta}
\def\lam{\lambda}
\def\th{\theta}
\def\s{\sigma}
\def\psidm{\psi_m}

\def\rc{random-cluster}

\def\eps{\epsilon}
\def\g{\gamma}
\def\Ga{\Gamma}
\def\Si{\Sigma}
\def\rc{random-cluster}
\def\ZZ{{\mathbb Z}}

\def\RR{{\mathbb R}}

\def\PP{{\mathbb P}}

\def\CC{{\mathbb C}}
\def\ZR{\ZZ\times\RR}
\def\Plb{\PP_{\lam,\de}}

\def\PLlb{\PP_{\La,\lam,\de}}
\def\PLlbq{\PP_{\La,\lam,\de,q}}
\def\PLmblb{\PP_{\Lamb,\lam,\de}}

\def\Om{\Omega}
\def\OmL{\Omega_\La}
\def\Ommb{\Om_{m,\b}}
\def\Omnb{\Om_{n,\b}}
\def\om{\omega}
\def\De{\Delta}

\def\Smb{\Sigma_{m,\b}}
\def\th{\theta}

\def\eps{\epsilon}
\def\La{\Lambda}
\def\Lamb{\La_{m,\b}}
\def\Lanb{\La_{n,\b}}
\def\oo{\infty}
\def\lest{\le_{\mathrm {st}}}

\def\thetac{\theta_{\mathrm{c}}}

\def\be{\begin{equation}}
\def\ee{\end{equation}}
\def\sm{\setminus}
\def\resp{respectively}
\def\pd{\partial}
\def\pdh{\pd^{\mathrm h}}
\def\lra{\leftrightarrow}
\def\nlra{\nleftrightarrow}
\def\es{\varnothing}
\def\phm{\phi_m}

\def\phmb{\phi_{m,\b}}
\def\ophmb{\ophi_{m,\b}}

\def\pn{\phi_n}
\def\ophm{\ophi_m}

\def\ophnb{\ophi_{n,\b}}
\def\ol#1{\overline{#1}}

\def\ophi{\ol\phi}

\def\blam{\bm{\lambda}}
\def\bde{\bm{\de}}
\def\Lam{\La_m}
\newcommand\an{a_n}

\begin{document}
\title[Bounded quantum entanglement entropy]{Bounded entanglement entropy\\ in the quantum Ising model}
\author[G. R. Grimmett, T. J. Osborne, P. F. Scudo]{Geoffrey R.\ Grimmett, Tobias J. Osborne, Petra F. Scudo}
\address{Centre for Mathematical Sciences,
University of Cambridge, Wilberforce Road, Cambridge CB3 0WB, UK}
\email{grg@statslab.cam.ac.uk}

\address{Institut f\"ur Theoretische Physik, Leibniz Universit\"at Hannover, 
Appelstr.\  2, 30167 Hannover, Germany}
\email{tobias.j.osborne@gmail.com}

\address{European Commission, Joint Research Centre,
Directorate B, Growth \&\ Innovation
Unit B6, Digital Economy
Via E.\ Fermi, 2749, 21027 Ispra (VA), Italy}
\email{pscudo@gmail.com}
   
\begin{abstract}
A rigorous proof is presented of the boundedness of the entanglement 
entropy of a block of
spins for the ground state of the one-dimensional quantum Ising
model with sufficiently strong transverse field.
This is proved by a refinement of the stochastic geometric arguments in the earlier
work by the same authors (J.\ Statist.\ Phys.\ 131 (2008) 305–339).
The proof utilises a transformation to a
model of classical probability called the continuum \rc\ model.
Our method of proof is fairly robust, and applies also to certain
disordered systems.  
\end{abstract}

\date{25 June 2019, revised 1 November 2019} 

\maketitle

\section{The quantum Ising model and entanglement}

The purpose of this note is to give
a rigorous proof of the area law for entanglement entropy
in the quantum Ising model in one dimension.
This is achieved by an elaboration of the
stochastic geometrical approach of \cite{GOS}.
We prove the boundedness of entanglement entropy of a block of spins
of size $L+1$ in the ground state of the model with sufficiently strong transverse field,
uniformly in $L$.
The current paper is presented as a development of the 
earlier work \cite{GOS} by the same authors,
to which the reader is referred for details of the background and basic theory.

The quantum Ising model in question is defined as follows.
We consider a block of $L+1$ spins in a line of length $2m+L+1$.
Let $L \ge 0$. For $m \ge 0$, let
$$
\Delta_m= \{-m,-m+1,\dots, m+L\}
$$
be a subset of the
one-dimensional lattice $\ZZ$, and attach to each vertex $x\in
\Delta_m$ a quantum spin-$\frac12$ with local Hilbert space
$\CC^2$. The Hilbert space $\mathcal{H}$ for the system is
$\mathcal{H} = \bigotimes_{x=-m}^{m+L} \CC^2$. A convenient basis
for each spin is provided by the two eigenstates
$|+\rangle=\left(\begin{matrix} 1\\0\end{matrix}\right)$,
$|-\rangle=\left(\begin{matrix}0\\1\end{matrix}\right)$, of the Pauli operator
$$
\sigma^{(3)}_x = \left(
\begin{array}{cc} 1 & 0
\\ 0 & -1\end{array} \right),
$$
at the site $x$, corresponding to the eigenvalues $\pm 1$.
The other two Pauli operators with respect
to this basis are represented by the matrices
\begin{equation}
\sigma^{(1)}_x= \left( \begin{array}{cc} 0 & 1 \\ 1 & 0\end{array}
\right), \qquad \sigma^{(2)}_x= \left ( \begin{array}{cc} 0& -i \\
i & 0\end{array}\right).
\end{equation}
A complete basis for $\mathcal{H}$ is given by the tensor products
(over $x$) of the eigenstates of $\sigma^{(3)}_x$. In
the following, $|\phi\rangle$ denotes a vector and $\langle \phi|$
its adjoint. As a notational convenience, we shall represent sub-intervals
of $\ZZ$ as real intervals, writing for example $\De_m=[-m, m+L]$.

The spins in $\Delta_m$ interact via the quantum Ising Hamiltonian
\begin{equation}\label{ham}
H_{m} = -\tfrac{1}{2} \sum_{\langle x,
y\rangle}\lambda\sigma^{(3)}_x
 \sigma^{(3)}_y -  \sum_{x} \delta\sigma^{(1)}_x,
\end{equation}
generating the operator $e^{-\b H_m}$ where $\b$ denotes inverse
temperature. Here,  $\lambda\geq 0$ and $\delta\geq 0$ are
the spin-coupling and external-field intensities, respectively, and
$\sum_{\langle x, y\rangle}$ denotes the sum over all (distinct)
unordered pairs of neighbouring spins. While we phrase our results for the
translation-invariant case, our approach can be extended to disordered systems with  couplings
and field intensities that vary across $\ZZ$, much as in \cite[Sect.\ 8]{GOS}. 
See Theorem \ref{thm5}.

The  Hamiltonian $H_m$ has a unique pure ground state $|\psidm
\rangle$ defined at zero temperature (as $\b\to\oo$) as the eigenvector
corresponding to the lowest eigenvalue of $H_m$. This ground state $|\psidm\rangle$ depends
only on the ratio $\theta=\lam/\de$.
We work here with a free boundary condition
on $\De_m$, but we note that the same methods are valid with a
periodic (or wired) boundary condition, in which $\De_m$ is embedded
on a circle. 

Write $\rho_m(\b)=e^{-\b H_m}/\tr(e^{-\b H_m})$, and
$$
\rho_m=\lim_{\b\to\oo}\rho_m(\b) =|\psidm \rangle\langle\psidm|
$$
for the density operator corresponding to the ground state of the
system.  The
ground-state entanglement of $|\psidm\rangle$ is quantified by
partitioning the spin chain $\Delta_m$ into two disjoint sets $[0,
L]$ and $\Delta_m\setminus [0, L]$ and by considering the entropy of
the \emph{reduced density operator} \be\label{reddo} \rho_m^L =
\tr_{\Delta_m\setminus [0, L]}(|\psidm\rangle\langle \psidm|). \ee
One may similarly define, for finite $\b$, the reduced operator
$\rho_m^L(\b)$. In both cases, the trace is performed over the
Hilbert space of spins belonging to
$\De_m\sm[0,L]$. Note that $\rho_m^L$ is a positive semi-definite
operator on the Hilbert space $\sH_L$ of dimension $d= 2^{L+1}$ of
spins indexed by the interval $[0, L]$. By the spectral theorem for
normal matrices \cite{bhatia}, this operator may be diagonalised and
has real, non-negative eigenvalues, which we denote in decreasing order by
$\lambda_j^{\downarrow}(\rho_m^L)$.

\begin{defn}\label{ent}
The \emph{entanglement (entropy)} of the interval
$[0,L]$ relative to its complement $\Delta_m \setminus [0, L]$ is given by
\begin{equation}\label{entdef}
S(\rho^L_m) = -\tr(\rho_m^L \log_2 \rho_m^L) =
-\sum_{j=1}^{2^{L+1}} \lambda_j^{\downarrow}(\rho_m^L)\log_2
\lambda_j^{\downarrow}(\rho_m^L),
\end{equation}
where $0 \log_2 0$ is interpreted as $0$.
\end{defn}

Here are our two main theorems.

\begin{thm}\label{mainest2}
Let $\lam,\de\in(0,\oo)$ and  $\th = \lam/\de$.
There
exists $C=C(\theta) \in(0,\oo)$, 
and a constant $\g=\g(\th)$ satisfying
$0<\g<\oo$ if $\th<2$, such that, for all $L\ge 1$,
\begin{equation}\label{eq:rcbound}
\|\rho_m^L-\rho_{n}^L\| \le \min\{2, C e^{-\gamma m}\},
\qquad 2\le m\le n.
\end{equation}
Furthermore, we may choose such $\g$ satisfying 
$\g(\th)\to\oo$ as $\th\downarrow 0$.
\end{thm}

Equation \eqref{eq:rcbound} is in terms of the operator norm:
 \beq\label{sup} 
 \| \rho^L_m- \rho^L_n
\|\equiv \sup_{\|\psi\|=1} \Big| \langle \psi |\rho^L_m-
\rho^L_n|\psi\rangle\Big|, 
\eeq 
where the supremum is taken over
all vectors $|\psi\rangle \in \sH_L$ with unit $L^2$-norm. 

\begin{remark}\label{rem:0}
The value $\theta=2$ is critical for the quantum Ising model in one dimension, 
and therefore the condition $\theta<2$ is sharp for $\gamma>0$ in \eqref{eq:rcbound}. 
See the discussion following \cite[Thm 7.1]{BjG}. 
\end{remark}

\begin{thm}\label{entest}
Consider the quantum Ising model \eqref{ham} on $n = 2m+L+1$ spins, with
parameters $\lam$, $\de$, and let $\g$ be as in Theorem \ref{mainest2}.
If
$\gamma > 2\ln 2 $, there exists $c_1=c_1(\theta)<\oo$ such that
\be\label{new50}
S(\rho_m^L) \le c_1,\qquad m,L \ge 0.
\ee
\end{thm}

Weaker versions of Theorems \ref{mainest2} and \ref{entest} were proved in \cite[Thms 2.2, 2.8]{GOS}, namely that \eqref{eq:rcbound}
holds subject to a power factor of the form $L^\alpha$, 
and \eqref{new50} holds with $c_1$ replaced by $C_1+C_2\log L$
(and subject to a slightly stronger assumption on $\g$). 
As noted in Remark \ref{rem:0}, Theorem \ref{mainest2} is a further strengthening of \cite[Thm 2.2]{GOS} in that
\eqref{eq:rcbound} holds for $\theta<2$, rather then just $\theta<1$. 
Stronger versions of these two theorems may be proved similarly,
with the interactions $\lam$ and field intensities $\delta$ varying with position while
satisfying a suitable condition.  A formal statement for the disordered case appears at
Theorem \ref{thm5}.

There is a considerable and growing  literature in the physics journals concerning entanglement entropy
in one and more dimensions. For example, paper \cite{ECP} is an extensive review of area laws.
The relationship between entanglement entropy and the spectral gap has been explored in 
\cite{AKLV,ALV}, and polynomial-time algorithms for simulating the ground state are studied in \cite{ALVV}.
Related works include studies of the XY spin chain \cite{AR},  oscillator systems \cite{BSW}, 
the XXZ spin chain \cite{BW}, and free fermions \cite{Pastur}. The connection between correlations and the area-law
is explored in \cite{BH}.

We make next some remarks about the proofs of the above two theorems.
The basic approach of these mathematically rigorous proofs
is via the stochastic geometric representation
of Aizenman, Klein, Nachtergaele, and Newman \cite{AKN, AN,BN}. Geometric techniques have proved of
enormous value in studying both classical systems (including Ising and Potts models, see for example \cite{G-RC}),
and quantum systems (see \cite{Bj15, Bj16, BjG, CI, GUW, BT93}). 

The proofs of Theorems \ref{mainest2}, \ref{entest} and the forthcoming Theorem \ref{thm5}
have much in common with
those of \cite[Thms 2.2, 2.8]{GOS} subject to certain improvements
in the probabilistic estimates. The general approach and many details are the same as in 
the earlier paper, 
and indeed there is some limited overlap of text. 
We  make frequent reference here to \cite{GOS}, 
and will highlight where the current proofs differ, while omitting
arguments that may be taken directly from \cite{GOS}.
In particular, the reader is referred to \cite[Sects.\ 4, 5]{GOS} for details of the
percolation representation of the ground state, and of the 
associated continuum random-cluster model. In Section \ref{be}, we review
the relationship between the  reduced density operator and the random-cluster model, and we
state the fundamental inequalities of Theorem \ref{mainest} and Lemma \ref{lem1}.
Once the last two results have been proved, Theorems \ref{mainest2} and \ref{entest} follow as in \cite{GOS}:
the first as in the proof of \cite[Thm 2.2]{GOS}, and the second as in that of \cite[Thm 2.8]{GOS}
(see the notes for the latter included in Section \ref{sec:proof2}).

We reflect in Section \ref{sec:disorder} on the extension of our methods and conclusions 
when the edge-couplings $\lam$ and field strengths $\de$ are permitted 
to vary, either deterministically or randomly, about the line. 
In this disordered case, the Hamiltonian \eqref{ham} is replaced by
\begin{equation}\label{ham2}
H_{m} = -\tfrac{1}{2} \sum_{\langle x,
y\rangle}\lambda_{x,y}\sigma^{(3)}_x
 \sigma^{(3)}_y -  \sum_{x} \delta_x\sigma^{(1)}_x,
\end{equation}
where the sum is over neighbouring pairs $\langle x,y\rangle$ of $\De_m$.
We write $\blam=(\lam_{x,x+1}: x \in \ZZ)$ and $\bde=(\de_x: x \in \ZZ)$.

\begin{thm}\label{thm5} 
Consider the quantum Ising model on $\ZZ$ with Hamiltonian \eqref{ham2}, such that, for some $\lam,\de>0$,
$\blam$ and $\bde$ satisfy
\be\label{eq:1005}
\lam_{x,y}/\de_x \le \lam/\de,\qquad y=x-1,x+1, \ x \in\ZZ.
\ee
\begin{letlist}
\item[\rm(a)] If $\lam/\de<2$, then \eqref{eq:rcbound} 
holds with $C$ and $\g$ as given there.
\item[\rm(b)] If, further, $\g> 2\ln 2$, then \eqref{new50} holds with $c_1$ as given there.
\end{letlist}
If $\blam$ and $\bde$ are random sequences satisfying \eqref{eq:1005} with probability one, then
parts {\rm(a)} and {\rm(b)} are valid a.s.
\end{thm}

The situation is more complicated when $\blam$, $\bde$ are 
random but do not a.s.\ satisfy \eqref{eq:1005} with $\lam/\de<2$.

\begin{remark}\label{rem1}
The authors acknowledge Massimo Campanino's announcement 
in a lecture on 12 June 2019 of
his perturbative proof with Michele Gianfelice
of a version of Theorem \ref{mainest2} for sufficiently small $\theta$,
using cluster expansions. That announcement stimulated the authors of the current work.
\end{remark}

\section{Estimates via the continuum random-cluster model}\label{be}

We write $\RR$ for the reals and $\ZZ$ for the integers.
The \emph{continuum percolation model} on $\ZR$ is constructed as in \cite{st-perc,GOS}. 
For  $x\in\ZZ$, let $D_x$ be a Poisson
process of points in $\{x\}\times\RR$ with intensity $\de$; the
processes $\{D_x: x\in \ZZ\}$ are independent, 
and the points in the $D_x$ are termed
`deaths'.  The lines $\{x\}\times \RR$ are called \lq time lines'.

For $x\in\ZZ$, let $B_x$ be a Poisson process of points
in $\{x+\frac12\}\times\RR$ with intensity $\lam$;
the processes $\{B_x: x\in\ZZ\}$ are independent of each other and of the
$D_y$. For $x\in\ZZ$ and each $(x+\frac12,t)\in B_x$, we
draw a unit line-segment in $\RR^2$ with endpoints $(x,t)$
and $(x+1,t)$, and we refer to this as a `bridge' joining
its two endpoints.  For $(x,s), (y,t) \in \ZZ\times\RR$, we write
$(x,s)\lra (y,t)$ if there exists a path $\pi$ in $\RR^2$ with
endpoints $(x,s)$, $(y,t)$ such that: $\pi$ comprises sub-intervals
of $\ZR$ containing no deaths, together possibly with bridges.
For $\La,\De \subseteq \ZR$, we write $\La\lra \De$ if there exist
$a\in \La$ and $b\in \De$ such that $a\lra b$.
Let $\PLlb$ denote the associated probability
measure when restricted to the set $\La$, and write $\theta=\lam/\de$.

Let $\Plb$ be the corresponding measure on the whole space $\ZR$, 
and recall from \cite[Thm 1.12]{BG} that the value $\th=1$ is the critical point of the continuum percolation model.

The \emph{continuum random-cluster model} on $\ZR$ is defined as follows.
Let $a,b\in\ZZ$, $s,t\in\RR$ satisfy $a \le b$ and $s \le t$, 
and write $\La=[a,b]\times[s,t]$ for the box
$\{a,a+1,\dots,b\} \times [s,t]$. Its boundary $\pd\La$
is the set of all points $(x,y)\in\La$
such that: either $x\in\{a,b\}$, or $y\in \{s,t\}$, or both.

As sample space we take the
set $\Om_\La$ comprising all finite subsets (of $\La$) of deaths and bridges,
and we assume that no death is the endpoint of any bridge.
For $\om\in\Om_\La$, we write $B(\om)$ and $D(\om)$ for the sets
of bridges and deaths, respectively, of $\om$. 

The \emph{top/bottom periodic boundary condition} is imposed on $\La$:
for $x\in [a,b]$, we identify the two points $(x,s)$
and $(x,t)$.  The remaining boundary of $\La$,
denoted $\pdh\La$, is the set of 
points of the form $(x,u)\in \La$ with $x\in\{a,b\}$ and $u\in [s,t]$. 
 
For $\om\in\Om_\La$, let $k(\om)$
be the number of its clusters, counted according to the connectivity relation $\lra$
(and subject to the above boundary condition). Let $q\in(0,\oo)$, and define the
`continuum \rc' probability measure
$\PLlbq$ by
\be\label{rcPo}
d\PLlbq(\om) = \frac 1Z q^{k(\om)}d\PLlb(\om),
\qq \om\in\Om_\La,
\ee
where $Z$ is the appropriate partition function.
As at \cite[eqn (5.3)]{GOS},
\be\label{stochcomp}
\PLlbq \lest \PLlb, \qq q \ge 1,
\ee
in the sense of stochastic ordering.

We introduce next a variant in which the box
$\La$ possesses a `slit' at its centre. Let $L\in\{0,1,2,\dots\}$ and
$S_L=[0,L]\times\{0\}$. We think of $S_L$ as a collection
of $L+1$ vertices labelled in the obvious way as
$x=0,1,2,\dots,L$. For $m\ge 2$, $\b>0$, let $\Lamb$ be the box
$$
\Lamb=[-m,m+L]\times[-\tfrac12\b,\tfrac12\b]
$$
subject to a `slit' along $S_L$. That is,
$\Lamb$ is the usual box except that each vertex  $x\in S_L$ is
replaced by two distinct vertices $x^+$ and $x^-$. The vertex
$x^+$ (\resp, $x^-$) is attached to the half-line
$\{x\}\times(0,\oo)$ (\resp, the half-line $\{x\}\times(-\oo,0)$);
there is no direct connection between $x^+$ and $x^-$. Write
$S_L^\pm=\{x^\pm: x\in S_L\}$ for the upper and lower sections of
the slit $S_L$. Henceforth we take $q=2$.
Let $\ophmb$ be the continuum \rc\ measure
on the slit box $\Lamb$ with  
parameters $\lam$, $\de$, $q=2$ and free boundary condition on $\pd\Lamb$,
and let $\phmb$ be the corresponding probability measure with top/bottom periodic boundary condition.

We make a note concerning exponential decay which will be important later.
The critical point of the infinite-volume ($q=2$) continuum \rc\ model
on $\ZR$ with parameters $\lam$, $\de$
is given by
$\thetac=2$ where $\th=\lam/\de$
(see \cite[Thm 7.1]{BjG}). Furthermore, as in \cite[Thm 5.33(b)]{G-RC},
there is a unique infinite-volume weak limit, denoted $\phi_{\lam,\de}$,
when $\theta<2$.
In particular (as in the discussion of \cite{BjG}) 
there is exponential decay of connectivity when $\theta<2$. Let $\La_m=[-m,m]^2 \subseteq \ZR$,
with boundary $\pd\La_m$.

\begin{thm}[\mbox{\cite[Thms 6.2, 7.1]{BjG}}]\label{contperc}
Let $\lam, \de \in(0,\oo)$, and $I=\{0\}\times[-\frac12,\frac12]\subseteq \ZR$. 
There exist $C=C(\lam,\de)\in(0,\oo)$
and $\g=\g(\lam,\de)$ satisfying $\g>0$ when $\th=\lam/\de<2$, such that
\begin{equation}\label{eq:new10}
\phi_{\lam,\de}\bigl(I\lra \pd\La_m\bigr) \le Ce^{-\g m},\qquad m\ge 0.
\end{equation}
The function $\g(\lam,\de)$ may be chosen to satisfy $\g\to\oo$
as $\de\to \oo$ for fixed $\lam$.
\end{thm}

Henceforth the function $\g$ denotes that of Theorem \ref{contperc}.
(The function $\g$ in Theorems \ref{mainest2}, \ref{entest} is derived from that
of Theorem \ref{contperc}.) By stochastic domination,
\eqref{eq:new10} holds with $\phi_{\lam,\de}$ replaced by $\PP_{\La,\lam,\de,2}$
for general boxes $\La$. 

It is explained in \cite{GOS} that a random-cluster configuration $\om$ 
gives rise, by a cluster-labelling process, to
an Ising configuration on $\La$, which serves (see \cite{AKN})  as a two-dimensional representation
of the quantum Ising model of \eqref{ham}.
We shall use
$\ophmb$ and $\phmb$ to denote  the respective couplings of the
continuum \rc\ measures and the corresponding (Ising) spin-configurations, and $\ophmb^\eta$, $\phmb^\eta$
for the measures with spin-configuration $\eta$ on $\pdh\Lamb$.

\begin{remark}\label{rem:3}
Theorem \ref{contperc} is an important component of the estimates that follow.
At the time of the writing of \cite{GOS}, the result was known only when $\theta<1$, 
and the corresponding exponential-decay theorem \cite[Thm 6.7]{GOS}  was proved by stochastic comparison with
continuum percolation (see \eqref{stochcomp}). More recent progress of \cite{BjG}
has allowed its extension to the $q=2$ continuum \rc\ model directly. 
In order to apply it in the current work, 
a minor extension of the ratio weak-mixing theorem \cite[Thm 7.1]{GOS} is needed,
namely that the mixing theorem holds with $\ol\phi$ taken to be the \rc\ measure on
$\La$ with \emph{free boundary conditions}. The proof is unchanged.
\end{remark}

\begin{remark}\label{rem:5}
In the proofs that follow, it would be convenient to have a stronger version of 
\eqref{eq:new10} with $\ophmb$ replaced by the finite-volume \rc\ measure on $\Lamb$ 
with wired boundary condition on $\pdh\Lamb$ and periodic top/bottom boundary condition.
It may be possible to derive such an inequality as in \cite{DRT}, but we do not pursue that option here.
\end{remark}

\begin{remark}\label{rem:6}
We shall work only in the subcritical phase $\theta=\lam/\de<2$. As remarked prior to Theorem 
\ref{contperc}, there exists a unique infinite-volume measure. Similarly, the limits 
\begin{equation}\label{eq:new66}
\ophm=\lim_{\beta\to\oo}\ophmb, \qq \phm=\lim_{\beta\to\oo}\phmb,
\end{equation}
exist and are identical measures on the strip $\Lam = [-m,m]\times(-\oo,\oo)$.
\end{remark}

Let $\Ommb$ be the sample space of the continuum \rc\ model
on $\Lamb$, and $\Smb$ the set of admissible allocations of spins to the clusters of
configurations, as in \cite[Sect.\ 5]{GOS}. For $\s\in\Smb$ and $x\in S_L$,
write $\s_x^\pm$ for the spin-state of $x^\pm$.
Let $\Si_L=\{-1,+1\}^{L+1}$ be the set of
spin-configurations of the vectors $\{x^+: x\in S_L\}$
and $\{x^-: x\in S_L\}$, and write $\s^+_L=
(\s_x^+: x\in S_L)$ and $\s^-_L=
(\s_x^-: x\in S_L)$.

Let
\be\label{eq:amdef2}
\amb=\ophmb(\s_L^+=\s_L^-).
\ee
Then,
\be\label{eq:35}
\amb\to a_m = \phm(\s_L^+=\s_L^-)\qq \text{as } \beta\to\oo,
\ee
where $\phm=\lim_{\beta\to\oo}\phmb$ as in Remark \ref{rem:6}.

Here is the main estimate of this section, of which Theorem \ref{mainest2}
is an immediate corollary with adapted values of the constants.
It differs from \cite[Thm 6.5]{GS} in the removal of a factor of order
$L^\alpha$, and the replacement of the condition $\theta<1$ by the weaker assumption $\theta<2$.

\begin{thm}\label{mainest}
Let $\lam, \de\in(0,\oo)$ and write $\th=\lam/\de$. 
If $\th < 2$, there exist $C,M \in(0,\oo)$, depending on
$\th$ only, such that
the following holds. For $L\ge 1$ and $M\le m\le n<\oo$,
\be\label{eq:36}
\sup_{\|c\|=1} \left| \frac{\phm(c(\s_L^+)c(\s_L^-))}{a_m}
- \frac{\pn(c(\s_L^+)c(\s_L^-))}{a_n} \right|
\le C  e^{-\frac13\g m},
\ee
where $\g$ is as in Theorem \ref{contperc}, and the supremum is over all functions $c:\Si_L\to\RR$ with 
$L^2$-norm satisfying $\|c\|=1$. 
\end{thm}

In the proof of Theorem \ref{mainest},
we make use of the following two lemmas
(corresponding, \resp, to \cite[Lemmas 6.8, 6.9]{GOS}), which
are proved in Section \ref{rwm} using the method of ratio weak-mixing.

\begin{lem}\label{lem1}
Let $\lam,\de\in(0,\oo)$ satisfy $\th=\lam/\de<2$, and let $\g$ be as in Theorem \ref{contperc}. 
There exist constants $A(\lam,\de),C_1(\lam,\de)\in(0,\oo)$ such that the following holds.
Let 
\be\label{eq:new52}
R_K = C_1 e^{-\frac12\g K}
\ee
For all $L\ge 3$, $1\le K<\frac12 L$, $m\ge 1$, $\b \ge 1$, and all $\eps^+,\eps^-\in\Si_L$, we have that
$$
A^{2K}(1-R_K) \le \frac{\ophmb(\s_L^+=\eps^+,\, \s_L^-=\eps^-)}
{\ophmb(\s_L^+=\eps^+)\ophmb(\s_L^-=\eps^-)}
\le A^{-2K}(1+R_K),
$$
whenever  $K$ is such that $R_K \le \tfrac12$.
\end{lem}

In the second lemma we allow a general spin boundary condition on $\pdh\Lamb$. 

\begin{lem}\label{lem2}
Let $\lam,\de\in(0,\oo)$ satisfy $\th=\lam/\de<2$, and let $\g$ be as in Theorem \ref{contperc}. 
There exists a constant $C_1\in(0,\oo)$ such that{\rm:}
for all $L\ge 3$, $m\ge 1$, $\b\ge 1$, all events $A\subseteq\Si_L\times\Si_L$,
and all admissible spin boundary conditions $\eta$
of $\pdh\Lamb$,
$$
\left| \frac{\ophmb^\eta((\s^+_L,\s_L^-)\in A)}
{\ophmb((\s^+_L,\s_L^-)\in A)}-1\right|
\le C_1e^{-\frac27\g m},
$$
whenever the right side of the inequality is less than $1$.
\end{lem}

\begin{proof}[Proof of Theorem \ref{mainest}]
Let $\theta<2$, and let $\g$ be as in Theorem \ref{contperc}.
It suffices to prove \eqref{eq:36} with $\phi_m$ (\resp, $\phi_n$) replaced by $\ophmb$
(\resp, $\ophnb$), and $a_m$ (\resp, $\an$) replaced by $\amb$ (\resp, $\anb$).
Having done so, we let $\b\to\oo$ to obtain \eqref{eq:36} by Remark \ref{rem:6}. 

Let $A$, $C_1$, $R_K$ be as in Lemma \ref{lem1}, and
let $L\ge 3$ and $1\le K<\frac12 L$ be such that
\be\label{new32}
R_K  \le \tfrac14.
\ee
Remaining small values of $L$ are covered in \eqref{eq:36} by adjusting $C$.

Since $\ophmb\lest\ophnb$, 
we may couple $\ophmb$ and $\ophnb$ via a probability measure
$\nu$ on pairs $(\om_1,\om_2)$ of configurations on $\Lanb$ in such a way that
$\nu(\om_1\le\om_2) = 1$. It is standard
(as in \cite{G-RC,New93}) that we may find $\nu$ such that
$\om_1$ and $\om_2$ are identical configurations within the region of $\Lamb$ that
is not connected to $\pdh\Lamb$ in the upper configuration 
$\om_2$.
Let $D$ be the set of all pairs 
$(\om_1,\om_2)\in\Omnb\times\Omnb$ such that: $\om_2$ contains
no path joining $\pd B$ to $\pdh\Lamb$, where
\begin{equation}\label{ruppbnd}
B = [-r,r+L]\times[-r,r], \qq r=\lfloor \tfrac12 m\rfloor.
\end{equation}
The relevant regions are illustrated in Figure \ref{fig:0}.

\begin{figure}[ht] 
\begin{center}
\includegraphics[width=0.7\textwidth]{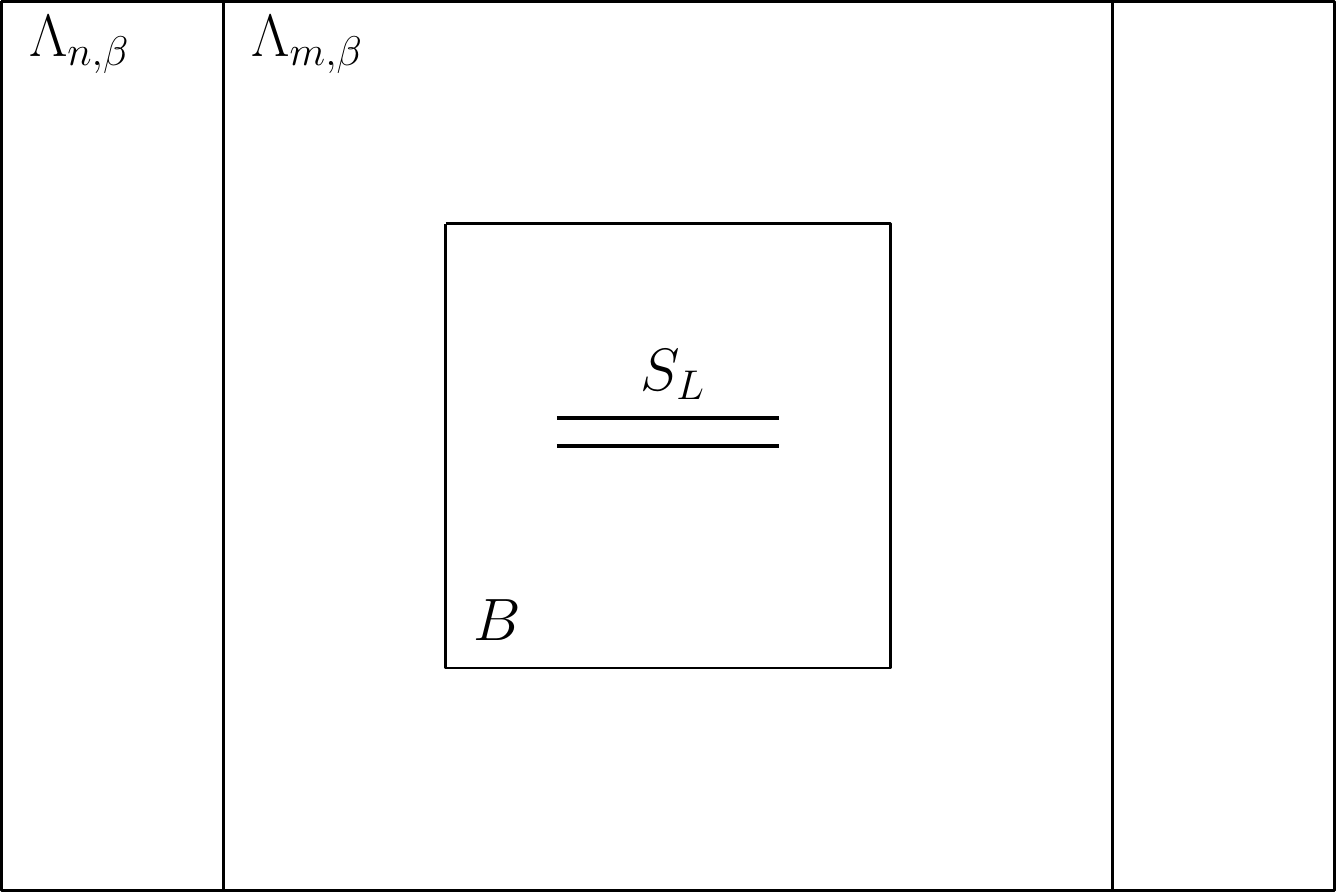}
\caption{The boxes $\Lanb$, $\Lamb$, and $B$.}\label{fig:0}
\end{center}
\end{figure}

Having constructed the measure $\nu$ accordingly, we may now allocate spins to the clusters of
$\om_1$ and $\om_2$ in the manner described in \cite[Sect.\ 5]{GOS}. This may be done in such a way that,
on the event $D$, the spin-configurations associated with $\om_1$ and $\om_2$
within $B$ are identical. We write $\s_1$ (\resp, $\s_2$) for
the spin-configuration on the clusters of $\om_1$ (\resp, $\om_2)$, and
$\s_{i,L}^\pm$ for the spins of $\s_i$ on the slit $S_L$.

By the remark following \cite[eqn (6.4)]{GOS}, it suffices to consider non-negative functions
$c:\Si_L\to\RR$, and thus we let  $c:\Si_L\to[0,\oo)$ with $\|c\|=1$. Let
\be
S_c = \frac{c(\s_{1,L}^+)c(\s_{1,L}^-)}{\amb}
- \frac{c(\s_{2,L}^+)c(\s_{2,L}^-)}{\anb},
\ee
so that
\be\label{e1}
\frac{\ophmb(c(\s_L^+)c(\s_L^-))}{\amb}
- \frac{\ophnb(c(\s_L^+)c(\s_L^-))}{\anb}
= \nu (S_c 1_D) + \nu(S_c 1_{\ol D}),
\ee
where $\ol D$ is the complement of $D$, and $1_E$ is the indicator function of $E$.

Consider first the term $\nu(S_c 1_D)$ in \eqref{e1}. On the event
$D$, we have that $\s_{1,L}^\pm = \s_{2,L}^\pm$, so that
\be\label{e2}
|\nu(S_c 1_D)| \le \left|1-\frac {\amb}{\anb}\right|
\frac{\ophmb( c(\s_{L}^+)c(\s_{L}^-))}{\amb}.
\ee

By Lemma \ref{lem1} and \cite[Lemma 6.10]{GOS},
\begin{align}
\ophmb(c(\s_{L}^+)c(\s_{L}^-))
&= \sum_{\eps^\pm \in \Si_L} c(\eps^+)c(\eps^-) \ophmb( \s_{L}^+=\eps^+,\, \s_{L}^-=\eps^- )
\nonumber\\
&\le A^{-2K}(1+R_K)  \ophmb(c(\s_L^+))\ophmb(c(\s_L^-))\nonumber\\
&=  A^{-2K}(1+R_K)  \left(\sum_{\eps\in \Si_L}c(\eps)\ophmb(\s_L^+=\eps) \right)^2\nonumber\\
&\le A^{-2K}(1+R_K) \sum_{\eps\in\Si_L} \ophmb(\s_L^+=\eps)^2,\label{eq:300}
\end{align}
where we have used reflection-symmetry in the horizontal axis at the 
intermediate step.
By Lemma \ref{lem1} and reflection-symmetry again, 
\begin{align*}
\amb &= \sum_{\eps\in\Si_L} \ophmb(\s_L^+=\s_L^-=\eps)\\
&\ge A^{2K}(1-R_K) \sum_{\eps\in\Si_L} \ophmb(\s_L^+=\eps)^2.
\end{align*}
Therefore, 
\be\label{e4}
\frac{\ophmb( c(\s_{L}^+)c(\s_{L}^-))}{\amb} \le
A^{-4K}\frac{1+R_K}{1-R_K}.
\ee

We set $A=\{\s_L^+=\s_L^-\}$ in Lemma \ref{lem2} to find that, for sufficiently
large $m\ge M_1(\lam,\de)$,
$$
\left| \frac{\ophmb^\eta(\s^+_L=\s_L^-)}
{\ophmb(\s^+_L=\s_L^-)}-1\right|
\le C e^{-\frac27 \g m}<\frac 12.
$$
Each of the two probabilities on the left side may be interpreted
as probabilities in the continuum Potts model of \cite[eqn (5.4)]{GOS} on $\Lam$.
By averaging over $\eta$, sampled according to $\ophnb$ when viewed as a Potts measure, we deduce 
by the spatial Markov property that
$$
\left| \frac{\ophnb(\s^+_L=\s_L^-)}
{\ophmb(\s^+_L=\s_L^-)}-1\right|
\le Ce^{-\frac27\g m} < \frac12,
$$
which is to say that
\be\label{e3}
\left|\frac{\anb}{\amb}-1\right| \le C e^{-\frac27\g m} < \frac12.
\ee

We make a note for later use.
In the same way as above, a version of inequality \eqref{e4} holds 
with $\ophmb$ replaced by the continuum \rc\ measure $\ophi_B$
on the box $B$ with free boundary conditions, namely,
\be\label{new53}
\frac{\ophi_B( c(\s_{L}^+)c(\s_{L}^-))}{a_B} \le
A^{-4K}\frac{1+R_K}{1-R_K},
\ee
where $a_B = \ophi_B(\s_L^+ = \s_L^-)$.
By \eqref{ruppbnd} and \eqref{e3}, we may take $C$ and $M_1$ above such that
\be\label{e10}
\left|\frac{\anb}{a_B}-1\right| \le C e^{-\frac17\g m} < \frac12, \qquad m\ge M_1(\lam,\de).
\ee

Inequalities \eqref{e4} and \eqref{e3} may be combined as in
\eqref{e2} to obtain
\be\label{e5}
|\nu(S_c 1_D)| \le  C_1 A^{-4K}\frac{1+R_K}{1-R_K} e^{-\frac27\g m}
\ee
for an appropriate constant $C_1= C_1(\lam,\de)$ and all $m\ge M_1$.

We turn to the term $\nu(S_c 1_{\ol D})$ in \eqref{e1}.
Evidently,
\be\label{e6}
|\nu(S_c 1_{\ol D})| \le A_m + B_n,
\ee
where
$$
A_m = \frac {\nu(c(\s_{1,L}^+)c(\s_{1,L}^-) 1_{\ol D})}{\amb},\q
B_n = \frac {\nu(c(\s_{2,L}^+)c(\s_{2,L}^-) 1_{\ol D})}{\anb}.
$$
There exist constants $C_2$, $M_2$ depending on $\lam$, $\de$,
such that, for $m>r \ge M_2$,
\begin{align}
B_n &=
\frac{\nu(\ol D)}{\anb} \nu(c(\s_{2,L}^+)c(\s_{2,L}^-)\mid \ol D )\nonumber\\
&= \frac{\nu(\ol D) }{\anb}
\ophnb\bigl( \ophi_B^\eta(c(\s_{2,L}^+)c(\s_{2,L}^-))\mid \ol D\bigr)
\nonumber\\ 
&\le \frac{\nu(\ol D)}{a_B} C_2 \ophi_B(c(\s_L^+)c(\s_L^-))
\label{eq:301}
\end{align} 
by Lemma \ref{lem2} with $\ophmb$ replaced by $\ophi_B$, and \eqref{e10}.
At the middle step, we have used conditional expectation
given the spin configuration $\eta$
on $\Lamb\sm B$.
By \eqref{new53}, 
\be\label{e7}
B_n\le \nu(\ol D) A^{-4K}\frac{1+R_K}{1-R_K}.
\ee

A similar upper bound is valid for $A_m$, on noting that the conditioning
on $\ol D$ imparts certain information about the configuration
$\om_1$ outside $B$ but nothing
further about $\om_1$ within $B$.
Combining this with \eqref{e6}--\eqref{e7}, we find that,
for $r \ge M_3(\lam,\de)$ and some $C_3=C_3(\lam,\de)$,
\be\label{e8}
|\nu(S_c 1_{\ol D})| \le \nu(\ol D) C_3 A^{-4K}\frac{1+R_K}{1-R_K}.
\ee
By \eqref{stochcomp}, \eqref{ruppbnd}, and Theorem \ref{contperc},
\be\label{e23}
\nu(\ol D) \le C_4 m e^{-\frac12\gamma m} \le C_5  e^{-\frac13\gamma m} ,
\qquad m \ge M_4,
\ee 
for some $C_4$, $C_5$, $M_4\ge 2M_3$.
We combine \eqref{e5}, \eqref{e8}, \eqref{e23} as in \eqref{e1}.
Letting $\b\to\oo$ and recalling \eqref{new32},
we obtain \eqref{eq:36} from \eqref{eq:35}, for $m\ge  M:=
\max\{M_1,M_2,M_4\}$.

Finally, we remark that $C$ and $M$ depend on both $\lam$ and $\de$.
The left side of \eqref{eq:36} is invariant under re-scalings of the time-axes,
that is, under the transformations $(\lam,\de) \mapsto (\lam\eta, \de\eta)$ 
for $\eta\in(0,\oo)$. We may therefore work with the new values
$\lam'=\th$, $\de'=1$, with appropriate constants $\a(\th,1)$,
$C(\th,1)$, $M(\th,1)$. 
\end{proof}

\section{Proofs of Lemmas \ref{lem1} and \ref{lem2}}\label{rwm}

Let $\La$ be a box in $\ZR$ (we shall later
consider a box $\La$ with a slit $S_L$, for which the same
definitions and results are valid).
A {\em path} $\pi$ of $\La$ is an alternating sequence of disjoint intervals
(contained in $\La$) and unit line-segments 
of the form $[z_0,z_1]$, $b_{12}$, $[z_2,z_3]$, $b_{34}$,
$\dots$, $b_{2k-1,2k}$, $[z_{2k},z_{2k+1}]$, where: each pair $z_{2i}$, $z_{2i+1}$ is on the
same time line of $\La$, and $b_{2i-1,2i}$ is a unit line-segment with endpoints $z_{2i-1}$
and $z_{2i}$, perpendicular to the time-lines. 
The path $\pi$ is said to join $z_0$ and $z_{2k+1}$.
The {\em length} of $\pi$ is its
one-dimensional Lebesgue measure. A \emph{circuit} $D$ of $\La$ is a path 
except inasmuch as
$z_0=z_{2k+1}$. A set $D$ is called \emph{linear}
if it is a disjoint union of paths and/or 
circuits. Let $\De$, $\Ga$ be disjoint subsets
of $\La$. The linear set $D$ is said to \emph{separate} $\De$ and $\Ga$
if every path of $\La$ from $\De$ to $\Ga$ passes through $D$, and $D$ is minimal with
this property in that no strict subset of $D$ has the property.

Let $\om\in\OmL$. An {\em open path} $\pi$ of $\om$ is a path of $\La$
such that, in the notation above,
the intervals $[z_{2i},z_{2i+1}]$ contain no death of $\om$, and the line-segments
$b_{2i-1,2i}$ are bridges of $\om$.

Let $\Ga$ be a measurable subset and $\De$ a finite subset of $\La$ such that
$\De\cap \Ga=\es$.
We shall make use of the `ratio weak-mixing property'
of the spin-configurations in $\De$ and $\Ga$ that is stated and proved in \cite[Thm 7.1]{GOS};
note Remark \ref{rem:3}.

Consider the box $\Lamb$ with slit $S_L$. Let $K$ be an integer satisfying
$1\le K < \frac12 L$, and let 
\be\label{new54}
\begin{aligned}
\De&= \{x^+: x\in S_L,\, K\le x\le L-K\},\\
\Ga &= \{x^-: x\in S_L,\, K\le x\le L-K\}.
\end{aligned} 
\end{equation}
The following replaces \cite[Lemma 7.24]{GOS}.

\begin{lem}\label{thm2}
Let $\lam,\de\in(0,\oo)$ satisfy $\th=\lam/\de<2$, and let  
$\g>0$ be as in Theorem \ref{contperc}. There exists  
$C_1=C_1(\lam,\de)\in(0,\oo)$ such that the following holds.
For $\eps_K^+\in\Si_\De$, $\eps_K^-\in\Si_\Ga$, we have that
$$
\left| \frac{\ophmb(\s_\De=\eps_K^+,\, \s_\Ga=\eps_K^-)}
{\ophmb(\s_\De=\eps_K^+)\ophmb(\s_\Ga=\eps^-_K)} - 1\right|
\le  C_1 e^{-\frac12\g K},
$$
whenever the right side is less than $\frac12$. 
\end{lem}

\begin{proof}
Take 
$$
D=\Bigl([-m,0)\times\{0\}\Bigr) \cup\Bigl((L,L+m])\times\{0\}\Bigr),
$$ 
the union of the two horizontal
line-segments that, when taken with the slit $S_L$, 
complete the `equator' of $\Lamb$.
Thus $D$ is a linear subset of $\Lamb$
that separates $\De$ and $\Ga$.
Let $t_1$, $t_2$, $t$ be as in \cite[Thm 7.1]{GOS}, namely,
\begin{equation}\label{eq:121}
\begin{gathered}
t_1=\ophmb(\De\lra D), \qq
t_2=\sqrt{\ophmb(D \lra \Ga)},\\
t= t_1+2t_2 +\frac{t_1+t_2}{1-t_1-2t_2}.
\end{gathered}
\end{equation}

By Theorem \ref{contperc}, there exist constants $C_2$, $C_3$, depending
on $\lam$ and $\de$ only, such that
\begin{align*}
t_1
\le 2\sum_{i=K}^{\lfloor L/2\rfloor} C_2 e^{-\g i}
\le C_3 e^{-\g K},
\end{align*}
and furthermore $t_2^2=t_1$.
The claim now follows by \cite[Thm 7.1]{GOS} and Remark \ref{rem:3}.
\end{proof}

We now prove Lemmas \ref{lem1} and \ref{lem2}. 

\begin{proof} [Proof of Lemma \ref{lem1}]
Let $\theta<2$ and let $\g$ be  as in Theorem \ref{contperc}. With $1 \le K < \frac12 L$,
write $\s_{L,K}^\pm=(\s_x^\pm: K \le x \le L-K)$.
First, let $x=(L,0)$, and let $\eps^+, \eps^-\in \{-1,+1\}^{L+1}$ be possible
spin-vectors of the sets $S_L^+$ and  $S_L^-$, respectively. By
\cite[Lemma 7.25]{GOS} with $S = S_L^+ \cup S_L^- \sm\{x^+\}$,
\begin{multline*}
\ophmb(\s_L^+ = \eps^+,\, \s_L^- =\eps^-)\\
\ge \tfrac12 \ophmb(\s_y^+=\eps_y^+\ \text{for}\ y\in S_L^+\sm\{x^+\},\, \s_L^-=\eps^-) 
\PLmblb(x^+ \nlra S).
\end{multline*} 
Now, $\PLmblb(x\nlra S)$ is at least as large as the probability that the
first event (death or bridge) encountered on moving northwards
from $x$ is a death, so that
$$
\PLmblb(x \nlra S) \ge  \frac{\de}{2\lam + \de}.
$$
On iterating the above, we obtain that
\be
\ophmb(\s_L^+=\eps^+,\, \s_L^-=\eps^-) \ge
A^{2K}
\ophmb(\s_{L,K}^+ = \eps_K^+,\, \s_{L,K}^- =\eps_K^-),
\ee
where $\eps^\pm_K$ is the vector obtained from $\eps^\pm$
by removing the entries labelled by vertices $x$ satisfying
$0\le x <K$ and $L-K < x \le L$, and
\be\label{new33}
A=\left(\frac \de{2(2\lam+\de)}\right)^2.
\ee
In summary,  for $\eps^\pm\in\Si_L$,
\begin{align}
A^{2K} \ophmb(\s_{L,K}^+=\eps_K^+,\, \s_{L,K}^-=\eps_K^-)
&\le \ophmb(\s_L^+=\eps^+,\ \s_L^-=\eps^-)\nonumber\\
&\le  \ophmb(\s_{L,K}^+=\eps_K^+,\, \s_{L,K}^-=\eps_K^-).
\label{new20}
\end{align}

With $\De$, $\Ga$ as in \eqref{new54},
we apply Lemma \ref{thm2}
to obtain that
there exists $C_1=C_1(\lam,\de)<\oo$ such that
\begin{equation}\label{new61}
\left|\frac{\ophmb(\s_{L,K}^+=\eps_K^+,\, \s_{L,K}^-=\eps_K^-)}
{\ophmb(\s_{L,K}^+=\eps_K^+)\ophmb(\s_{L,K}^-=\eps_K^-)}-1\right|
\le C_1 e^{-\frac12\g K},
\end{equation}
whenever the right side is less than or equal to $\frac12$.

By a similar argument to \eqref{new20},
\begin{equation}\label{new62}
A^K \ophmb(\s_{L,K}^\pm=\eps_K^\pm)
\le \ophmb(\s_L^\pm=\eps^\pm)
\le  \ophmb(\s_{L,K}^\pm=\eps_K^\pm).
\end{equation}
The claim follows on combining \eqref{new20}--\eqref{new62}.
\end{proof}

\begin{proof}[Proof of Lemma \ref{lem2}]
Let $\De=S_L^+ \cup S_L^-$ and $\Ga=\pdh\Lamb$, and suppose  $\theta<2$. Let $k=\frac37 m$ and
assume for simplicity that $k$ is an integer. (If either $m$ is small or $k$ is
non-integral, the constant $C$ may be adjusted accordingly.)
Let $D_0$ be the circuit illustrated in Figure \ref{fig:4}, comprising a path in
the upper half-plane from $(-k,0)$ to $(L+k,0)$ together with its reflection
in the $x$-axis. Let $D=D_0 \cap \Lamb$. Thus, $D=D_0$ in the case $\beta=\beta_2$ of the figure.
In the case $\b=\b_1$, $D$ comprises two disjoint paths of $\Lamb$. 
In each case, $D$ separates $\Delta$ and $\Sigma$.

\begin{figure}[t] 
\psfrag{x}[][]{$x_2$}
\psfrag{y}[][]{$x_3$}
\begin{center}
\includegraphics[width=0.6\textwidth]{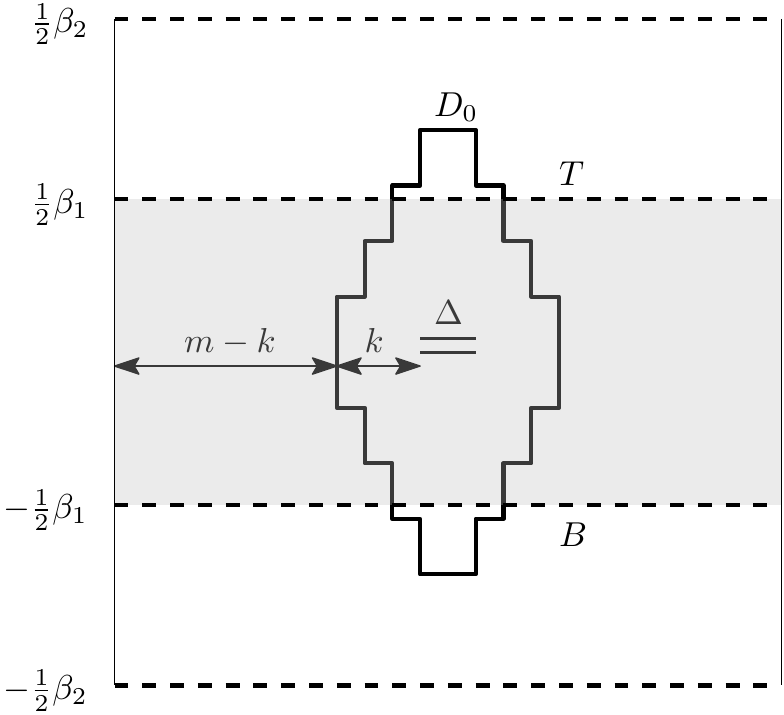}
\caption{The circuit $D_0$ is approximately a parallelogram with $\De$
at its centre. The sides comprise vertical steps of 
height $2$ followed by horizontal
steps of length 1. The horizontal and vertical diagonals of $D_0$ have lengths $2k+L$ and
(approximately) $4k+2L$ respectively, where $k=\frac37 m$. 
Two values of $\beta$ are indicated. When $\beta=\beta_2$, $D_0$ is contained in $\Lamb$
and we take $D=D_0$.
When $\beta=\beta_1$, $\Lamb$ is the shaded area only, and we work with 
$D=D_0 \cap \Lamb$ considered as the union of two disjoint paths that separates $\De$ and $\Sigma$.}
\label{fig:4}
\end{center}
\end{figure}

Let $t_1$, $t_2$, $t$ be as in \eqref{eq:121}.
By the ratio weak-mixing theorem \cite[Thm 7.1]{GOS} and Remark \ref{rem:3},
$$
\left| \frac{\ophmb^{\eta}((\s^+_L,\s_L^-)=(\eps^+,\eps^-))}
{\ophmb((\s^+_L,\s_L^-)=(\eps^+,\eps^-))}-1\right|
\le 2t,\qq \eps^\pm\in\Si_L,
$$
whenever $t\le \frac12$.
We multiply up, and sum over $(\eps^+,\eps^-)\in A$
to obtain 
\be\label{eq:203}
\left| \frac{\ophmb^{\eta}(\s_\De\in A)}
{\ophmb(\s_\De\in A)}-1\right|
\le 2t,
\ee
whenever $t \le \frac12$.

By Theorem \ref{contperc}, there exist $C_2,C_3,c_4>0$, depending on $\lam$, $\de$, such that
\begin{align}\nonumber
t_1 &\le 4\sum_{i=0}^{\lfloor L/2\rfloor} \ophi((i,0)\lra D_0)\\
&\le 4\sum_{i=0}^{\lfloor L/2\rfloor}  C_2 e^{-\g\frac23( k +i)} \le C_3 e^{-\frac27\g m},
\label{eq:2001}\end{align}
and similarly,
\be\label{eq:2002}
t_2^2 \le 8\sum_{i=0}^{\lceil k+L/2\rceil} C_2 e^{-\g(\frac47m+ c_4 i)} \le C_3 e^{-\frac47\g m}.
\ee
The claim follows.
\end{proof}

\section{Quenched disorder}\label{sec:disorder}

The parameters $\lam$ and $\de$ have so far been assumed constant. The situation
is more complicated in the disordered case, when either they vary deterministically,
or they are random. The
arguments of this paper may be applied in both cases, and the outcomes 
are summarised in this section.
Let the Hamiltonian \eqref{ham} be replaced by \eqref{ham2},
and write $\blam=(\lam_{x,x+1}: x \in \ZZ)$ and $\bde=(\de_x: x \in \ZZ)$.

The fundamental bound of Theorem \ref{mainest} depends only on
the ratio $\th=\lam/\de$. In the disordered setting, the connection probabilities
of the continuum \rc\ model are increasing in $\blam$
and decreasing in $\bde$, and powers of the function $A(\lam,\de)$ of \eqref{new33} are 
replaced by products of the form
\be\label{new33a}
A'_{x,k}=\prod_{i=1}^k \left(\frac {\de_{x+i}}{2(\de_{x+i}+\lam_{x+i,x+i-1}+\lam_{x+i,x+i+1})}\right),
\ee
which are decreasing in $\blam$
and increasing in $\bde$.
By examination of the earlier lemmas and proofs, 
the conclusions of the paper are found to be valid with $\g=\g(\lam,\de)$ whenever
\eqref{eq:1005} holds with some $\lam,\de>0$.
Hence, in the disordered case where \eqref{eq:1005}
holds with probability one, the corresponding conclusions are valid a.s.\ (subject to appropriate
bounds on the ratio $\lam/\de$). This proves Theorem \ref{thm5}.

Consider now the situation in which \eqref{eq:1005} does not
hold with probability one. Suppose that the $\lam_{x,x+1}$, $x\in\ZZ$,
are independent, identically distributed random variables, and
similarly the $\de_x$, $x\in\ZZ$, and assume that the vectors $\blam$ and $\bde$ 
are independent. 
We write $P$ for the corresponding probability
measure, viewed as the measure governing the `random environment'.

A quenched area law might assert something along the following lines:
subject to  suitable conditions, there exists 
a random variable $Z$ which is $P$-a.s.\ finite such that
$S(\rho_m^L)<Z$ for all appropriate $m$, $L$.
Such a uniform upper bound will not generally exist, owing to the fluctuations in the system as $L \to\oo$.
In the absence of an assumption of the type of \eqref{eq:1005}, there may exist sub-domains of 
$\ZZ$ where the environment is not propitious for such a
bound. 

Partial progress may be made using the methods of \cite[Sect.\ 8]{GOS}, but this
is too incomplete for inclusion here.

\section{Proof of Theorem \ref{entest}}\label{sec:proof2}

Since this proof is very close to that of \cite[Thm 2.12]{GOS}, we include only details 
that are directly relevant to 
the strengthened claims of the current theorem, namely the removal of the logarithmic term
of \cite{GOS} and the weakened assumption on $\g$.
 
Let $C$ and $\g>2\ln 2$ be
as in Theorem \ref{mainest2}, and  choose an integer $K =K(\theta)\ge 2$ such that 
\begin{equation}\label{eq:new67}
 C e^{-\g K} \le 1.
\end{equation}
As in \cite{GOS}, 
\be\label{eq:1002}
S(\rho_m^L)  \le 2K, \qq 2 \le m \le K,
\ee
and we assume henceforth that $m>K$.

Let $\eps(r)= C e^{-\gamma (K+r)}$, so that, by \eqref{eq:new67},
\be\label{eq:eps}
\eps(r) \le e^{-\g r},\qquad r \ge 0.
\ee
On following the proof of \cite[Thm 2.8]{GOS} up to equation (2.22) there, we find that
\be\label{eq:eigbounds}
\lambda_j^{\downarrow}(\rho_{m}^L) \le \frac{c}{ j^\xi},
\qquad 2^{2K} < j,
\ee
where $\xi = {\gamma}/(2\ln 2) > 1$ and $c=e^{\gamma(K+1)}/(1-e^{-\g})$.

Now,
\be\label{eq:1003}
S(\rho_m^L) =  S_1 + S_2,
\ee
where 
$$
S_1 = -\sum_{j=1}^{\nu} \lambda_j^{\downarrow}(\rho_m^L)
\log_2 \lambda_j^{\downarrow}(\rho_m^L), \quad
S_2 =
-\sum_{j=\nu+1}^{2^{L+1}} \lambda_j^{\downarrow}(\rho_m^L)
\log_2 \lambda_j^{\downarrow}(\rho_m^L),
$$ 
and $\nu = 2^{2(K+2)}$.
Since the $\lambda_j^{\downarrow}(\rho_m^L)$, $1\le j\le \nu$, are non-negative 
with sum $Q$ satisfying $Q\le 1$, we have
\begin{equation}\label{eq:103}
S_1 \le \log_2 \nu =2(K+2).
\end{equation}

We use  \eqref{eq:eigbounds}
to bound $S_2$ as in \cite{GOS}, to obtain
\begin{align*}
S_2 &\le 
 -\sum_{j=\nu+1}^{\infty} \frac{c}{ j^{\xi}}\log_2 \left(\frac{c}{j^{\xi}}\right) 
\le c_1, 
\end{align*}
for some $c_1=c_1(\theta)<\oo$.
By \eqref{eq:1003}--\eqref{eq:103},
\be\label{eq:1004}
S(\rho_m^L) \le 2(K+2) + c_1, \qq m \ge K,
\ee
which completes the proof.

\providecommand{\bysame}{\leavevmode\hbox to3em{\hrulefill}\thinspace}
\providecommand{\MR}{\relax\ifhmode\unskip\space\fi MR }
\providecommand{\MRhref}[2]{%
  \href{http://www.ams.org/mathscinet-getitem?mr=#1}{#2}
}
\providecommand{\href}[2]{#2}

\end{document}